\documentclass[11pt]{article}
\RequirePackage[tt=false, type1=true]{libertine}
\RequirePackage[varqu]{zi4}
\RequirePackage[T1]{fontenc}
\usepackage[numbers,sort]{natbib}

\let \originalleft \left
\let\originalright\right
\renewcommand{\left}{\mathopen{}\mathclose\bgroup\originalleft}
\renewcommand{\right}{\aftergroup\egroup\originalright}

\usepackage{color,array,balance}
\usepackage{fullpage}
\usepackage{paralist}
\usepackage{booktabs}
\usepackage{style}

\usepackage{microtype}

\title{Parallelism in Randomized Incremental Algorithms}

\author{
  Guy E. Blelloch\\CMU\\guyb@cs.cmu.edu \and
  Yan Gu\\CMU\\yan.gu@cs.cmu.edu \and
  Julian Shun\\MIT CSAIL\\jshun@mit.edu \and
  Yihan Sun\\CMU\\yihans@cs.cmu.edu
}
\date{}

\begin{document}

\maketitle

\begin{abstract}
  In this paper we show that many sequential randomized incremental
  algorithms are in fact parallel.  We consider algorithms for several
  problems including Delaunay triangulation, linear programming,
  closest pair, smallest enclosing disk, least-element lists, and
  strongly connected components.

We analyze the dependences between iterations in an algorithm, and
show that the dependence structure is shallow with high probability,
or that by violating some dependences the structure is shallow and the
work is not increased significantly.  We identify three types of
algorithms based on their dependences and present a framework for
analyzing each type.  Using the framework gives work-efficient
polylogarithmic-depth parallel algorithms for most of the problems
that we study.

This paper shows the first incremental Delaunay triangulation algorithm with optimal work and polylogarithmic depth, which is an open problem for over 30 years. This result is important since most implementations of parallel Delaunay triangulation use the incremental approach. Our results also improve bounds on strongly connected components and least-elements lists, and significantly simplify parallel algorithms for several problems.
\end{abstract}

\section{Introduction}
The randomized incremental approach has been a very
useful
paradigm for generating simple and efficient algorithms for a variety
of problems.  There have been many dozens of papers on the topic (e.g., see
the surveys~\cite{Seidel93,Mulmuley94}).  Much of the early work was
in the context of computational geometry, but the approach has also been
applied to graph algorithms~\cite{cohen1997,coppersmith2003}.  The main idea is to insert
elements one-by-one in random order while maintaining a desired
structure.  The random order ensures that the insertions are somehow spread
out, and worst-case behaviors are unlikely.

The incremental process appears sequential since it is iterative,
but in practice incremental algorithms are widely used in parallel
implementations by allowing some iterations to start in parallel and
using some form of locking to avoid conflicts.  Many parallel
implementations for Delaunay triangulation and convex hull, for
example, are based on the randomized incremental
approach~\cite{Cignoni,Diaz04,Cintra2004,Llanos05,BBK06,Gonzalez2006,Stamp08,Pingali11,SBFG}.
In theory, however, there are still no known bounds
for parallel Delaunay triangulation using the incremental approach,
nor for many other problems.

In this paper we show that the incremental approach for Delaunay
triangulation, and
many other problems, is indeed parallel and leads to work-efficient
polylogarithmic-depth (time) algorithms for the problems.  The results
are based on analyzing the dependence graph (more accurately the
distribution of dependence graphs over the random order).  This
technique has recently been used to analyze the parallelism available
in a variety of sequential algorithms, including the simple greedy
algorithm for maximal independent set~\cite{BFS12}, the Knuth shuffle
for random permutation~\cite{SGBFG2015}, greedy graph
coloring~\cite{hasenplaugh2014ordering}, and correlation
clustering~\cite{PanPORRJ15}.  The advantage of this method is that
one can use standard sequential algorithms with modest change to make
them parallel, often leading to very simple parallel solutions.  It
has also been shown experimentally that the incremental approach leads to
practical parallel algorithms~\cite{BFGS}, and to deterministic
parallelism~\cite{Bocchino09,BFGS}.

\begin{table*}[t]
\centering
\def\arraystretch{1.1}
\begin{tabular}{c@{ }@{ }c@{ }@{ }c@{ }@{ }c}
\toprule
Problem & Work  & Depth & Type \\ \midrule
Comparison sorting (Section~\ref{sec:sorting})
& $O(n\log n)$ & $O(\log n)$ & 1\\
%
%
%
Delaunay triangulation, $d$ dims. (Section~\ref{sec:delaunay})
& $O(n \log n + n^{\lceil d/2 \rceil})\dagger$
& $O(d \log n \log^* n)$ & 1\\
2D linear programming (Section~\ref{sec:lp})
& $O(n)\dagger$ & $O(\log n)$ & 2\\
2D closest pair (Section~\ref{sec:cp})
& $O(n)\dagger$ & $O(\log n\log^*n)$ & 2\\
Smallest enclosing disk (Section~\ref{sec:sed})
& $O(n)\dagger$ & $O(\log^2n)$ & 2\\
Least-element lists (Section~\ref{sec:LE-lists})
& $O(W_{\smb{SP}}(n,m)\log n)\dagger$ & $O(D_{\smb{SP}}(n,m)\log n)$
& 3\\
Strongly connected components (Section~\ref{sec:SCC})
& $O(W_{\smb{R}}(n,m)\log n)\dagger$
& $O(D_{\smb{R}}(n,m)\log n)$ & 3\\
\bottomrule
\end{tabular}
\vspace{5pt}
\caption{Work and depth bounds for our parallel randomized incremental
  algorithms. $W_{\smb{SP}}(n,m)$ and $D_{\smb{SP}}(n,m)$ denote the
  work and depth, respectively, of a single-source shortest paths
  algorithm. $W_{\smb{R}}(n,m)$ and $D_{\smb{R}}(n,m)$ denote the work
  and depth, respectively, of performing a reachability query.  Bounds
  marked with a $\dagger$ are expected bounds, and the rest are high
  probability bounds.   All bounds are for the arbitrary  CRCW PRAM,
  except for comparison sorting that requires the priority CRCW
  PRAM. In all cases the work is the same as the sequential
  incremental algorithm, since the algorithms are effectively
  equivalent beyond either reordering (Type 1 or 2) or some redundancy
  (Type 3).}
\label{table:bounds}
\vspace{-10pt}
\end{table*}

\medskip

The contributions of the paper can be summarized as follows.

\begin{enumerate}[\hspace{-3pt} 1.]\itemsep 2pt 
\item We describe a framework for analyzing parallelism in randomized
  incremental algorithms.  We consider three types of dependences
  (Type 1, 2, and 3), and give general bounds on the depth of
  algorithms with for each type (Section~\ref{sec:incalg}).

\item
We show that randomly ordered insertion into a binary search tree is
inherently parallel, leading to an almost trivial comparison sorting algorithm
taking $O(\log n)$ depth and $O(n \log n)$ work (i.e., $n$ processors),
both with high probability on the priority-write CRCW PRAM (Section~\ref{sec:sorting}).
Surprisingly, we know of no previous description and analysis of this
parallel algorithm.

\item
  We show that an offline variant of Boissonnat and
  Teillaud's~\cite{BT93} randomized incremental algorithm  for
  Delaunay triangulation in $d$ dimensions has dependence depth $O(d \log n)$ with high
  probability (Section~\ref{sec:delaunay}).  We then describe a
  way to parallelize the algorithm, which leads to a parallel version with
  $O(d \log n \log \log n)$ depth with high probability, and $O(n \log n
  + n^{\lceil d /2 \rceil})$ work in expectation, on the CRCW PRAM.
  This is the first incremental construction of Delaunay triangulation with optimal work and polylogarithmic depth.  This problem has been open for 30 years, and is important since most implementations of parallel
Delaunay triangulation use the incremental approach, but none of them have polylogarithmic depth bounds.
  Surprisingly, our algorithm is very simple.

\item
We show that classic sequential randomized incremental algorithms for
constant-dimensional linear programming, closest pair, and smallest
enclosing disk have shallow dependence depth
(Section~\ref{sec:type2}).  This leads to very simple linear-work and
polylogarithmic-depth randomized parallel algorithms for all three
problems.

\item
  We show that by relaxing dependences (i.e., allowing some to be
  violated), two random incremental graph algorithms have (reasonably)
  shallow dependence depth.  The relaxation increases the work, but
  only by a constant factor in expectation.  We apply the approach to
  generate efficient parallel versions of Cohen's algorithm~\cite{cohen1997} for
  least-element lists  (Section~\ref{sec:LE-lists}) and Coppersmith
  et al.'s algorithm~\cite{coppersmith2003} for strongly connected
  components (SCC, Section~\ref{sec:SCC}).  In both cases we improve on the
  previous best bounds for the problems.
  Least-element lists have applications to tree embeddings on graph
  metrics~\cite{FRT,kanat}, and estimating neighborhood sizes in
  graphs~\cite{cohen2004}.  Coppersmith et al.'s SCC algorithm~\cite{coppersmith2003}
  is widely used in practice~\cite{hong2013fast,barnat2011computing,slota_ipdps2014_bfs,Tomkins2015}.
  In this paper, we analyze the parallelism of this algorithm, which had been a long-standing open question. This algorithm was later implemented and shown experimentally to be practical by Dhulipala et al.~\cite{dhulipala2018theoretically}.\footnote{We thank Laxman Dhulipala for catching a mistake in the original conference version of this paper when he was implementing the algorithm.  We have fixed the mistake in this version.}
\end{enumerate}

Other than the graph algorithms, which call subroutines that are known
to be hard to efficiently parallelize (reachability and
shortest paths), all of our solutions are work-efficient and run in
polylogarithmic depth (time).  The bounds for all of our parallel
randomized incremental algorithms can be found in
Table~\ref{table:bounds}. 


\subsubsection*{Preliminaries} We analyze parallel algorithms in the
work-depth paradigm~\cite{JaJa92}.  An algorithm proceeds in a
sequence of $D$ (depth) rounds, with round $i$ doing $w_i$ work in
parallel.  The total work is therefore $W = \sum_{i=1}^D w_i$.  We
account for the cost of allocating processors and compaction in our
depth.  Therefore the bounds on a PRAM with $P$ processors is $O(W/P +
D)$ time~\cite{brent1974parallel}.  We use the concurrent-read and
concurrent-writes (CRCW) PRAM model.  By default, we assume the
arbitrary-write CRCW model, but when stated use the priority-write
model.  We say $O(f(n))$ \defn{with high probability (\whp{})} to indicate
$O(kf(n))$ with probability at least $1- 1/n^k$.

\newcommand{\hp}{\hat{p}}
\newcommand{\hpij}{\hat{p}_{ij}}
\newcommand{\pij}{p_{ij}}

\section{Iteration Dependences}
\label{sec:incalg}
An \defn{iterative algorithm} is an algorithm that runs in a sequence
of \defn{\step{}s} (steps) in order.  When applied to a
particular input, we refer to the computation as an \defn{iterative
  computation}.
\Step{} $j$ is said to \defn{depend} on \step{} $i < j$ if
the computation of \step{} $j$ is affected by the computation of
\step{} $i$.  The particular dependences, or even the number of
\step{}s, can be a function of the input, and can be modeled as a
directed acyclic graph (DAG)---the \step{}s ($I = {1, \ldots, n}$) are
vertices and dependences between them are arcs (directed edges).

\begin{definition}[\IDG~\cite{SGBFG2015}]
An \textbf{\idg{}} for an
iterative computation is a (directed acyclic) graph $G(I,E)$ such that
if every \step{} $i \in I$ runs after all predecessor \step{}s in $G$
have completed, then every \step{} will do the same computation as in the
sequential order.
\end{definition}

We are interested in the depth (longest directed path) of \idg{} since
shallow dependence graphs imply high parallelism---at least if the
dependences can be determined online, and depth of each \step{} can be
appropriately bounded.  We refer to the depth of the DAG as the
\defn{\idd}, and denote it as $\idepth{G}$.

In general there can be \substep{}s nested within each \step{} of an
algorithm.  In this case we can consider the dependences between these
\substep{}s instead of the top-level \step{}s (i.e., a dependence from
the \substep{} in one \step{} to the \substep{} in either the same or
different \step{}).  The \idg{} is defined analogously---dependence
edges go between the \substep{}s, possibly in different top-level
\step{}s.  In this paper, we only consider one such algorithm,
Delaunay triangulation, where the main \step{}s are over the points,
and the \substep{}s are for each triangle created by adding the point.

An \emph{incremental algorithm} takes a sequence of elements (or
objects) $E$, and iteratively inserts then one at a time while
maintaining some property over the elements.  A \defn{randomized
  incremental algorithm} is an incremental algorithm in which the
elements are added in a uniformly random order---each permutation is
equally likely.  In this paper, we are interested in deriving
probability bounds over the \idd{}.
We consider three types of randomized incremental algorithms, which we
refer to as Type 1, 2, and 3, for lack of better names.

\subsection{Type 1 Algorithms}
\label{sec:incalg1}

In these algorithms we show the
probability bounds on \idd{} by considering all possible
paths of dependences.  By bounding the probability of each path, and
bounding the number of possible paths, the union bound can be used to
bound the probability that any path is long.
We use backwards analysis~\cite{Seidel93} to analyze the length
and number of paths.

We say that an incremental algorithm has \emph{$k$-bounded
  dependences} if for any input $E$ and element $e \in E$ inserted
last, $e$ directly depends on at most $k$ other elements---i.e., once
those up to $k$ other elements are inserted, it is safe to insert $e$.
For example, consider sorting by inserting into a binary search tree based on a random order.
For any key $v$ inserted last, once the previous and next keys in
sorted order, $v_p$ and $v_n$, have been inserted, we can immediate
insert $v$.
In particular the key $v$ will either be the right child
of $v_p$ or the left child of $v_n$ depending on which of the two was
inserted later.  More subtly, the search path for $v$ will also be the
same once $v_p$ and $v_n$ are both inserted (more discussion in
Section~\ref{sec:sorting}).  Inserting into a BST therefore has
$2$-bounded dependences.

If the \step{}s in an incremental algorithm are nested, then we
consider the pairs of an element along with each of its \substep{}s.
We say that an incremental algorithm has \emph{$k$-bounded nested
  dependences} if for any input $E$, any element $e \in E$ inserted
last, and any \substep{} $s$ for $e$, $(e,s)$ directly depends on at
most $k$ possible previous element--\substep{} pairs.  We say ``possible''
here since the \substep{}s might differ depending on the order of the
previous elements.  For example, consider Delaunay triangulation in
$d$ dimensions.  Inserting an element (point) $x$ will run \substep{}s
adding a set of triangles ($d$-simplices).  As shown in
Section~\ref{sec:delaunay}, each new triangle (\substep{}) will depend
on at most two previous triangles.  Each of these previous triangles
could have been added by a \substep{} of any of its $(d+1)$ corner
points, whichever was inserted last.  Hence there are $2 (d + 1)$
possible element--\substep{} pairs that the \substep{} for $x$ could
depend on, and Delaunay triangulation therefore has $2(d+1)$-bounded
nested dependences.

For a given insertion order of all elements, a \emph{tail} is an
element (or element--\substep{} pair for the nested case) that no
other element (or element-\substep{} pair) depends on.  The \emph{tail
  count} is the number of possible tails over all orderings.
Inserting $n$ keys into a BST, for example, has a tail count of $n$
since every key can be a tail.  For Delaunay triangulation every final
triangle can be involved in a tail, and each one created by any of its
corners, depending on which corner is last.  Therefore the tail count is at
most the number of final triangles times $(d+1)$.

We are now interested in the length of dependence paths for
incremental algorithms with $k$ bounded dependences, either nested or
not, and how that limits the \idd{}.

\begin{theorem}
  \label{thm:type1}
  Consider a random incremental algorithm on $n$ elements with
  $k$-bounded (nested) dependences, and for which the tail count is
  bounded by $c n^b$, for  some  constants  $b $ and  $c$.  Over the distribution of \idg{}s $G$, and
  for all $\sigma \geq k e^2$:
\[\Pr(D(G) \geq \sigma H_n) < c n^{-(\sigma - b)} \]
where $H_n = \sum_{i=1}^n 1/i$.
\end{theorem}

\begin{proof}
  We use backwards
  analysis by considering removing elements one-by-one from the last iteration.
  We analyze a specific dependence path to a tail and then take a union
  bound.

  Consider one of the possible $cn^b$ tails.  It corresponds to a
  single element $e$.  Starting at the end, the probability of the
  event ``element $e$ is at iteration $i = n$'' is $1/n$ since all
  permutations are equally likely.  If $e$ is at $i$, then $i$ is on a
  dependence path for that tail.  We then arbitrarily choose one of the
  up to $k$ remaining elements that $e$ depends on (or \substep{}s of
  an element for the nested case).  Lets now call it $e$---i.e., the
  element we are looking for is always called $e$.  Now we move back
  to $i = n - 1$, and the probability $e$ is at $i = n -1$ is again
  $1/i$.  This is true whether we found our original $e$ at $n$ or
  not---in both cases we are looking for a single element out of $i$
  possible elements.  Repeating this process until $i=1$, the
  probability that element $e$ is at iteration $i$ is always at most
  $1/i$ for all $i$.  Each time $e$ is at $i$ we extend the dependence
  path by $1$ and make another arbitrary choice among $k$
  predecessors, updating $e$ with our choice.

  Let $l$ be a random variable corresponding to the total number of
  dependences on the path we are considering.  We therefore have that
  $E[l] \leq \sum_{i=1}^n \frac{1}{i} = H_n$.  Furthermore each event
  ($e$ at $i$) is independent, giving the Chernoff bound:

    \[ P[l > \sigma E[l]] < \left(\frac{e^{\sigma-1}}{\sigma^{\sigma}}\right)^{E[l]}
  < \left(\frac{e}{\sigma}\right)^{\sigma E[l]} . \]

  We now take a union bound over the $cn^b$ possible tails and the at
  most $k^l$ possible choices we make for a predecessor for a
  dependence path of length $l$.  For $\sigma \geq k e^2$, we have:
  \begin{eqnarray*}
    P[D(g) > \sigma H_n]
    & \leq & c n^b k^{l} \cdot P[l > \sigma H_n] \\
    & < & c n^b k^{\sigma H_n} \left(\frac{e}{\sigma}\right)^{\sigma H_n} \\
    & = & c n^b \left(\frac{k e}{\sigma}\right)^{\sigma H_n} \\
    & \leq & c n^b \left(\frac{1}{e}\right)^{(\ln n) \sigma} \\
    & = & c n^{-(\sigma -b)}~.
  \end{eqnarray*}
\end{proof}

The Type 1 algorithms that we describe can be parallelized by running
a sequence of rounds.  Each round checks all remaining \step{}s to see
if their dependences have been satisfied and runs the \step{}s if so.
They  can be implemented in two ways: one
completely online, only seeing a new element at the start of each
\step{}, and the other offline, keeping track of all elements from the
beginning.  In the first case, a structure based on the history of all
updates can be built during the algorithm that allows us to
efficiently locate the ``position'' of a new element
(e.g.,~\cite{GKS92}), and in the second case the position of each
uninserted element is kept up-to-date on every \step{}
(e.g.,~\cite{CS89}).  The bounds on work are typically the same in
either case.  Our incremental sort uses an online style algorithm,
and the Delaunay triangulation uses an offline one.

\subsection{Type 2 Algorithms}
\label{sec:incalg2}

Type 2 incremental algorithms have a
special structure.  The \idg{} for these
algorithms is formed as follows: each \step{} $j$ is either a \emph{special
\step{}} or a \emph{regular \step{}} (depending on insertion order and the
particular element).  Each special \step{} $j$ has dependence arcs to all
\step{}s $i < j$, and each regular \step{} has one dependences arc to the
closest earlier special \step{}.  The first \step{} is special.  Furthermore
the probability of being a special \step{} is upper bounded by $c/j$ for
some constant $c$ and independent of the choices $j+1$ to $n$.  For
Type 2 algorithms, when a special \step{} $i$ is processed, it will check
all previous \step{}s, requiring $O(i)$ work and depth denoted as $d(i)$, and when a
non-special \step{} is processed it does $O(1)$ work.

\begin{theorem}\label{thm:type2}
  A Type 2 incremental algorithm has an \idd{} of
  $O(\log n)$ \whp{}, and can be implemented to run in $O(n)$ expected
  work and $O(d(n)\log n)$ depth \whp{}, where $d(n)$ is the depth of
  processing a special \step{}.
\end{theorem}

\begin{proof}
Since the probability of a special \step{} is bounded by $c/j$
independently of future \step{}s, the expected number of special
\step{}s is $\sum_{j=1}^nc/j = O(\log n)$, and using a Chernoff bound,
the number of special \step{}s is $O(\log n)$ \whp{}.  By construction
there cannot be more than two consecutive regular \step{}s in a path
of the \idg{}, so the \idd{} is at most twice the
number of special \step{}s and hence $O(\log n)$.

We now show how parallel linear-work implementations can be
obtained. A parallel implementation needs to execute the special \step{}s
one-by-one, and for each special \step{} it can do its computation in
parallel. For the non-special \step{}s whose closest earlier special \step{}
has been executed, their computation can all be done in parallel.
To maintain work-efficiency, we cannot afford to keep all unfinished
\step{}s active on each round.  Instead, we start with a constant number
of the earliest \step{}s on the first round and on each round
geometrically increase the number of \step{}s processed, similar to the
prefix methods described in earlier work on parallelizing iterative
algorithms~\cite{BFS12}.

\newcommand{\deref}[1]{^*\!#1}
\begin{algorithm}[t]
\caption{Type 2 Algorithm}\label{alg:type2}
\SetKwFor{ParForEach}{parallel foreach}{do}{endfch}
\KwIn{\Step{}s $[0,\ldots,n)$.}
\medskip
run special \step{} $0$\\
$j \gets 1$\\
\For {\upshape $i \leftarrow 2$ to $\log_2 n$} {
  \While {$j < 2^{i-1}$} {
    \ParForEach {$k \in [j, \ldots, 2^{i-1})$} {
        $F[k]\gets $ check if \step{} $k$ is special}
    $l \gets$ minimum true index in $F$, or $2^{i-1}$ if none\\
    \ParForEach {$k \in [j, \ldots, l)$} {
        run regular \step{} $k$}
    \If{ $l < 2^{i-1}$}{run special \step{} $l$}
    $j \gets l$}}
\end{algorithm}

Pseudocode is given in Algorithm~\ref{alg:type2}.  Without loss of
generality, assume $n=2^k$ for some integer $k$.  We refer to the
outer \textbf{for} loop as rounds, and the inner \textbf{while} loop
as sub-rounds.  Each round $i$ processes \step{}s
$[2^{i-2}, \ldots, 2^{i-1}]$ which we refer to as a \emph{prefix}.   The
variable $j$ at the start of each sub-round indicates that all \step{}s
before $j$ are done, and all \step{}s at or after are not.  Each sub-round
finds the first unfinished special \step{} $l$ within the round, if any.
It then runs all regular \step{}s up to $l$ (all their dependences are
satisfied).  Finally, if a special \step{} was found, that special \step{} is
run (all its dependences are satisfied).  Finding the first unfinished
special \step{} requires a minimum, which can be computed in
$O(2^i)$ work and $O(1)$ depth \whp{} on an arbitrary CRCW
PRAM~\cite{Vishkin10}.    Running all regular \step{}s also requires
$O(2^i)$ work and $O(1)$ depth, and running the special \step{} requires
$O(2^i)$ work and $O(d(n))$ depth.
The number of sub-rounds within a round is one more than
the number of special \step{}s in the prefix, which for any prefix $k$ is
bounded by $\sum_{i=2^{k-2}}^{2^{k-1}-1}c/i = O(1)$ in expectation.
Therefore, the work per round is $O(2^{i-1})$ in expectation, and
summed over all rounds is $O(1)+\sum_{i=2}^{\log n}O(2^{i-1}) = O(n)$
in expectation.     The number of sub-rounds is bounded by the number
of special \step{}s plus $\log n$, and each sub-round has depth $O(d(n))$
so the total depth is
$O(d(n)\log n)$ \whp{}.
\end{proof}

\subsection{Type 3 Algorithms}
\label{sec:incalg3}

In the third type of incremental algorithms it is safe to run \step{}s
in parallel, but this can require extra work.  In these algorithms an
\step{} can ``separate'' future \step{}s.  A simple example, again, is
insertion into a binary search tree, where the first key inserted
separates keys less than it from ones greater than it.  The idea is to
then process the \step{}s in rounds of increasing powers of two, as in
the Type 2 case.  However in this case, every \step{} in a round will
run as if it is at the beginning of the round ignoring conflicts, and
conflicts are resolved after the round.  We apply this approach to two
graph problems: least-elements (LE) lists and strongly connected
components (SCC).

\newcommand{\less}{<}
\newcommand{\sepdep}{separating dependences}

Consider a set of elements $S$.  We assume that each element $x \in S$
defines a total ordering $\less_x$ on all $S$.  This ordering can be
the same for each $x \in S$, or different.  For example, in sorting
the total ordering would be the order of the keys and the same for all
$x \in S$.  For a DAG, the ordering could be a topological sort, and
possibly different for each vertex since topological sorts are not
unique.  For both our applications, LE-lists and SCC, the orderings
can be different for each element.  The distinct orders is the
innovative aspect of our analysis.

\begin{definition}[\sepdep]
  \label{def:sepdep}
An incremental algorithm has \defn{\sepdep} if for all input $S$ (1) it
has total orderings $\less_x, x \in S$, and (2) for any three elements $a,
b, c \in S$, if $a \less_c b \less_c c $ or $c \less_c b \less_c a$,
then $c$ can only depend on $a$ if $a$ is inserted first among the
three.
\end{definition}
In other words, if $b$ separates $a$ from $c$ in the total ordering for $c$, and runs first, it will
separate the dependence between $a$ and $c$ (also if $c$ runs before $a$, of course,
there is no dependence from $a$ to $c$).   Again, using sorting as an
example, if we insert $b$ into a BST first (or use it as a pivot in
quicksort), it will separate $a$ from $c$ and they will never be
compared (each comparison corresponds to a dependence).
Let $d_{(i,j)}$ be the event that there is a dependence from \step{}
$i$ to \step{} $j$, and $p(d_{(i,j)})$ be its probability over all insertion orders.

\begin{lemma}
  \label{lem:sep}
In a randomized incremental algorithm that has \sepdep{},
we have $$p(d_{(i,j)}\,|\,d_{(i+1,j)},\ldots,d_{(j-1,j)} ) \leq 2/i$$ for $1\le i<j\le n$.
\end{lemma}

\begin{proof}
  Consider the total ordering $\less_j$.  Among the elements inserted in the first $i$
  \step{}s, at most two of them are the closest (by $\less_j$) to
  the element inserted at \step{} $j$ (at most one on each side).
  There will be a dependence from \step{} $i$ to $j$ only if the
  element selected on \step{} $i$ is one of these two---otherwise
  \step{}s before $i$ would have separated $i$ from $j$.   Since all
  of the first $i$ elements are equally likely to be selected on \step{} $i$, and
  this is independent of choices after $i$, the
  conditional probability is at most $2/i$.
\end{proof}

\begin{corollary}
The number of dependences in a randomized incremental algorithm with
\sepdep{} is $O(n \log n)$ in expectation.\footnote{\small Also true \whp{}.}
\end{corollary}
This comes simply from the sum $\sum_{j=2}^n\sum_{i=1}^{j-1} p(d_{(i,j)})$ which
is upper bounded by $2 n \ln n$.
This leads, for example, to a proof that quicksort, or randomized
insertion into a binary search tree, does $O(n \log n)$ comparisons in
expectation.  This is not the standard proof based on $p_{ij} =
2/(j-i+1)$ being the probability that the $i$'th and $j$'th smallest
elements are compared~\cite{CLRS}.  Here the $p_{ij}$ represent
the probability that the $i$'th and $j$'th elements in the random
order are compared.

In this paper, we introduce graph algorithms that have \sepdep{} with
respect to the processing order of the vertices, and there is a dependence from
vertex $i$ to vertex $j$ if a search from $i$ (e.g., shortest path or
reachability) visits $j$.

\begin{algorithm}[t]
\caption{Type 3 Algorithm}\label{alg:type3}
\SetKwFor{ParForEach}{parallel foreach}{do}{endfch}
\KwIn{\Step{}s $[0,\ldots,n)$.}
\medskip
run \step{} $0$\\
\For {$i \leftarrow 1$~~\textbf{\emph{to}}~$\log_2 n$} {
  \ParForEach {$k \in [2^{i-1}, \ldots, 2^i)$} {
    Run \step{} $k$ as if at \step{} $2^{i-1}$, i.e., using the final state from
    the previous round
  }
  \ParForEach {$k \in [2^{i-1}, \ldots, 2^i)$} {
    Combine state such that earlier $k$ have higher priority \label{line:combine}\\
    (Final state should be the same as if run sequentially up to $2^i -1$)
    }
    }
\end{algorithm}

To allow for parallelism, we permit \step{}s to run concurrently in
rounds, as shown in Algorithm~\ref{alg:type3}.  This means that we
might not separate \step{}s that were separated in the sequential
order.  For example, if $a$ separates $b$ from $c$ ($a < b < c$) in
the sequential order, but we run $a$ and $b$ in the same round, then
$c$ might depend on $b$ in the parallel order.  We therefore have to
consider $b$ as running at the start of the round (position $2^{i-1}$)
in determining the probability $p(d_{b,c})$.  This will cause
additional work, but as argued in the theorem below, the work is only increased by a
constant factor.  The second parallel loop is needed to combine
results from the \step{}s that are run in parallel.  The technique
here depends on the algorithm, but is simple for the algorithms we
consider for LE-Lists and SCC.

Consider applying the approach to insertion into a binary search tree.
On each round $i$, $2^{i-1}$ keys are already inserted into a BST and
in parallel we try to insert the next $2^{i-1}$ keys.  In the first
loop all new keys will search the tree for where they belong.  Many
will fall into their own leaf and be happy, but there will be some
conflicts in which multiple keys fall into the same leaf.  The second
loop would resolve these conflicts.   This is a different parallel algorithm
than the Type 1 algorithm described in Section~\ref{sec:sorting}.

We say that iteration $a$ has a left (right) dependence to a later
iteration $b$ if $b$ depends on $a$ and $a <_b b$ ($b <_b a$).
This definition is used to show the total number of dependences
of a specific iteration as follows.

\begin{lemma}
  \label{lem:type3}
  When applying Algorithm~\ref{alg:type3} to an incremental algorithm
  with \sepdep{}, let $p_{ij}(l), j\ge 2^i$ be the probability that
  $l$ iterations in round $i$ have a left dependence to iteration $j$.
  Then for all $i$ and $j$, we have $p_{ij}(l) \leq 2^{-l}$.
\end{lemma}
\begin{proof}
  Clearly $p_{ij}(0) \leq 2^{-0} = 1$.
    The probability that among iterations $[0,\ldots,2^i)$,
    the closest iteration to $j$ based on $<_j$ appears among $[2^{i-1},\ldots,2^i)$ is $1/2$ (since elements are in random order).  Therefore $p_{ij}(1) \le 1/2$.  Now given
  that the first is closest, the probability that the second is
  closest out of the remaining iterations in $[2^{i-1},\ldots,2^i)$ is $(2^{i-1} - 1)/(2^i - 1) < 1/2$.  Hence,  the probability for $l=2$ is less than $1/4$.  This repeats so
  $p_{ij}(l) < 2^{-l}$ for $l > 1$, giving our bound.
\end{proof}
We can make the symmetric argument about dependences on the right.
Importantly the expected number of dependences from a round to a later element
is constant, and the probability that the number of dependences is large is low.

\begin{theorem}
  \label{thm:type3}
  A randomized incremental algorithm with \sepdep{} can run in $O(\log
  n)$ parallel rounds over the \step{}s and every \step{} will
  have $O(\log n)$ incoming dependences \whp{} (for a total of
  $O(n \log n)$ \whp{}).
\end{theorem}

\begin{proof}
  We just consider left dependences, the right ones will just double the count.
  For fixed $j$ the upper bounds on the
  probabilities $p_{ij}$ are independent across the rounds $i$.
  This is because
  working backwards each round picks a random set of $1/2$ the remaining
  elements.
  The round that iteration $j$ belongs to contains less dependences than previous rounds, and the rounds later have no dependences to iteration $j$.
  Therefore, $p_{ij}(l) \leq 2^{-l}$ holds for all rounds even when $j<2^i$.

  For a set of independent random variables $X_i$ with exponential
  distribution $X_i \sim Exp(a)$, the sum $X = \sum X_i$
  satisfies the following tail bounds~\cite{janson2018tail}:
  \[P(X > \sigma \mathbb{E}[X]) \leq \sigma e^{-a\mathbb{E}[X](\sigma -1 - \ln \sigma)}\]
  For $\log_2 n$ rounds, $\mathbb{E}[X] = (\log_2 n) / a$ and $P(X >
  (\sigma/a) \log_2 n) \leq \sigma n^{-(\sigma -1 - \ln \sigma)}$ which
  satisfies the high probability condition.  Since by
  Lemma~\ref{lem:type3} our distributions are sub-exponential, the
  tails are no larger.
\end{proof}
Theorem~\ref{thm:type3} does not explicitly give the work and depth
for an algorithm since it will depend on the costs of running each
\step{}.  These will be given for the particular algorithms in
Section~\ref{sec:graph}.

\newcommand{\incsort}{\textsc{IncrementalSort}}

\section{Comparison Sorting (Type 1)}
\label{sec:sorting}

We first consider how to use our framework for sorting by
incrementally inserting into a binary search tree (BST) with no
rebalancing.  For simplicity we assume no two keys are equal.
It is well-known that for a random insertion order, inserting into a
BST takes $O(n \log n)$ time (comparisons) in expectation, or even with high
probability.  We apply our Type 1 approach to show that the sequential
incremental algorithm is also efficient in parallel.
Algorithm~\ref{alg:incsort} gives pseudocode that works either
sequentially or in parallel.  An \step{} is one round of the
\textbf{for} loop on Line 2.  For the parallel version, the
\textbf{for} loop should be interpreted as a \textbf{parallel for},
and the assignment on Line~\ref{line:assign} should be considered a
priority-write---i.e., all writes happen synchronously across the $n$
\step{}s, and when there are writes to the same location, the earliest
iteration gets written.  The sequential version does not need the check on
Line~\ref{line:par} since it is always true.

\newcommand{\deref}[1]{^*\!#1}
\begin{algorithm}[t]
\caption{\incsort}\label{alg:incsort}
\KwIn{A sequence $K = \{k_1,\ldots,k_n\}$ of keys.}
\KwOut{A binary search tree over the keys in $K$.}
  \vspace{.3em}
\tcp{\textrm{$\deref{P}$ reads indirectly through the pointer $P$.\\
The check on Line~\ref{line:par} is only needed for the parallel version.}}
  \vspace{.3em}
Root $\gets$ a pointer to a new empty location\\
\For {$i \leftarrow 1$ to $n$} {
  $N\gets$ newNode($k_i$)\\
  $P\gets$ Root\\
  \While {true} {
    \If {$~\deref{P} =$ \emph{null} \label{line:isempty}} {
      write $N$ into the location pointed to by $P$ \label{line:assign}\\
      \If {$~\deref{P} = N$ \label{line:par}}{break~~~~~~~// write succeeded and iteration $i$ is done}}
    \If {$N$.key $<{} \deref{P}$.key ~} {
        $P\gets$ pointer to $\deref{P}$.left}
    \Else {$P\gets$ pointer to $\deref{P}$.right}
  }
}
\Return {\emph{Root}}
\end{algorithm}

The dependence between \step{}s in the algorithm is in the check if
$\deref{P}$ is empty in Line~\ref{line:isempty}.  This means that \step{}
$j$ depends on $i < j$ if and only if the node for $i$ is on the path
to $j$.    The only important dependence is the last one on the path, since
all the others are subsumed by the last one (i.e. they do not appear in the
transitive reduction of the dependence graph).


\begin{lemma}
  Insertion of $n$ keys into a binary search tree in random order has
  \idd{} $O(\log n)$ \whp{}.
\end{lemma}

\begin{proof}
  When inserting an element at \step{} $i$ (removing in backwards
  analysis), there are at most two keys it can directly depend on, the
  previous and the next in sorted order (it is at most since there
  might not be a previous or next key).  This is among the keys from
  \step{}s $1$ to $i-1$.  Therefore there is a $2$-bounded dependence for
  all \step{}s.  Every key can be a tail (a leaf in the final binary
  search tree), so the tail count is $n$.  Using
  Theorem~\ref{thm:type1} we therefore have that the \step{} depth
  is bounded by $\sigma H_n$ for $\sigma > 2 e^2$ with probability at
  most $n^{-\sigma + 1}$.
\end{proof}

We note that since \step{}s only depend on the path to the key, the
transitive reduction of the \step{} dependence graph is simply the
BST itself.  In general, e.g. Delaunay triangulation in the next
section, the dependence structure is not a tree.


\begin{theorem}
  The parallel version of \incsort{} generates the same tree as
  the sequential version, and for a random order of $n$ keys runs in
  $O(n \log n)$ work and $O(\log n)$ depth \whp{} on a priority-write
  CRCW PRAM.
\end{theorem}
\begin{proof}
  They generate the same tree since whenever there is a dependence,
  the earliest \step{} wins.  The number of rounds of the while loop is
  bounded by the \idd ($O(\log n)$ \whp{}) since for each
  \step{}, each round checks a new dependence (i.e., each round
  traverses one level of the \idg).  Since each round
  takes constant depth on the priority-write CRCW PRAM with $n$
  processors, this gives the required bounds.
\end{proof}

Note that this gives a much simpler work-optimal logarithmic-depth
algorithm for comparison sorting than Cole's mergesort
algorithm~\cite{Cole88}, although it is on a stronger model
(priority-write CRCW instead of EREW) and is randomized.

\newcommand{\repboundary}{\textsc{ReplaceBoundary}}
\newcommand{\seqinc}{\textsc{IncrementalDT}}
\newcommand{\parinc}{\textsc{ParIncrementalDT}}
\newcommand{\incircle}{\textsc{InCircle}}
\newcommand{\gt}{G_T}
\newcommand{\tri}{t}
\newcommand{\points}{V}
\newcommand{\point}{v}
\newcommand{\dsimplex}{triangle}
\newcommand{\dsimplexes}{triangles}
\newcommand{\R}{\mathbb{R}}

\section{Delaunay Triangulation (Type 1)}\label{sec:delaunay}

A Delaunay triangulation (DT) in $d$ dimensions is a triangulation of
a set of points $P$ in $\R^d$ such that no point in $P$ is inside the
\emph{circumsphere} of any triangle (the sphere defined by the
triangle's $d+1$ corners).  Here we will use \emph{triangle} to mean a
$d$-simplex defined by $d+1$ corner points, and use \emph{face} to
mean a $d-1$ simplex with $d$ corner points.  We say a point
\defn{encroaches} on a triangle if it is in the triangle's
circumsphere, and will assume for simplicity that the points are in
general position, i.e., no $k \leq d + 1$ points on a
$(k-2)$-dimensional hyperplane, or $k \leq d + 2$ points on a
$(k-2)$-dimensional sphere.  Delaunay triangulation for $d=2$ can be
solved sequentially in optimal $O(n \log n + n^{\lceil d/2 \rceil})$
work.  There are also several work-efficient (or near work-efficient)
parallel algorithms for $d=2$ that run in polylogarithmic
depth~\cite{RS92b,Cole96,BHMT99}, and at least one for higher
dimension~\cite{AmatoGR94}, but they are all complicated.


The widely-used and simple incremental Delaunay algorithms date back
to the 1970s~\cite{GrSi78}.  They are based on the rip-and-tent idea:
for each point $p$ in order, rip out the triangles $p$ encroaches on
and tent over the resulting cavity with triangles from $p$ to each
boundary face of the cavity.  The algorithms differ in how the
encroached triangles are found, and how they are ripped and tented.
Clarkson and Shor~\cite{CS89} first showed that randomized incremental
convex hull is efficient, running with work $O(n \log n + n^{\lfloor
  d/2 \rfloor})$ in expectation.  These results imply optimal $O(n
\log n + n^{\lceil d/2 \rceil})$ work for DT.


Guibas et al.~\cite{GKS92} (GKS) showed a simpler direct randomized
incremental algorithm for 2d DT with optimal bounds, and this has
become the standard version described in
textbooks~\cite{Mulmuley94,BCKO,Edelsbrunner06} and often used in
practice.  The GKS algorithm uses a history of triangle updates to
locate a triangle $\tri$ that a new point $p$ encroaches.  It then
searches out for all other encroached triangles flipping pairs of
triangles as it goes.  Edelsbrunner and Shah~\cite{ES96} generalized
the GKS method to work in arbitrary dimension with optimal work (again
in expectation).  The algorithms, however, are inherently sequential
since for certain inputs and certain points in the input, the search
from $\tri$ will likely have depth $\Theta(n)$, and hence a single
\step{} can take linear depth.

Boissonnat and Teillaud~\cite{BT93} (BT) consider a someone different
but equally simple direct random incremental algorithm for DT that
does optimal work in arbitrary dimension.  Instead of using the
history to locate a single triangle that a point $p$ encroaches and
then searching out from it for the rest, it locates all encroached
triangles directly using the history.  It therefore does not suffer
the inherent sequential bottleneck of GKS.

\textbf{Our result.}  Here we show that an offline variant of the BT
algorithm has \idd{} $O(d \log n)$ \whp{}.  We further show
that the \step{}s can be parallelized leading to a very simple parallel
algorithm doing no more work than the sequential version (i.e.,
optimal), and with overall $O(d \log^2 n)$ depth \whp{}.

\SetKwInput{KwMaintains}{Maintains}%
\begin{algorithm}[!t]
\caption{\seqinc}\label{alg:delaunayseq}
\KwIn{A sequence $\points = \{\point_1,\ldots,\point_n\}$ of points in $R^d$.}
\KwOut{DT($\points$).}
\KwMaintains{A set of \dsimplexes{} $M$, and for each $\tri \in M$, the
  points that encroach on it, $E(\tri)$}
 \vspace{.5em}
$\tri_b \gets$ a sufficiently large bounding \dsimplex{}\\
$E(\tri_b) \gets \points$\\
$M \gets \{t_b\}$ \\
\For {$i \leftarrow 1$ to $n$} {
     Let $R\gets \{t\in M~|~v_i\in E(t)\}$\\
     \ForEach {\upshape face $f$ on the boundary of $R$} {
       $(\tri,\tri_o) \gets$ the two \dsimplex{}s incident on $f$, with
       $\tri$ on the $v_i$ side\\
       \repboundary{}$(t_o,f,t, \point_i)$
    }
}
\Return {$M$}

  \vspace{.5em}
    \SetKwProg{myfunc}{function}{}{}
    \myfunc{\upshape\repboundary{}$(t_0,f,t, \point)$} {
      $\tri' \gets$ a new \dsimplex{} consisting of $f$ and $v$\\
      \mbox{$E(\tri') \gets \{v' \in E(\tri) \cup E(\tri_o)~|~\mbox{\incircle{}}(v',t')\}$}\label{line:test}\\
      detach $\tri$ from face $f$ in $M$\\
      add $\tri'$ to $M$
    }

\end{algorithm}

\begin{figure}
\begin{center}
  \includegraphics[width=.3\columnwidth]{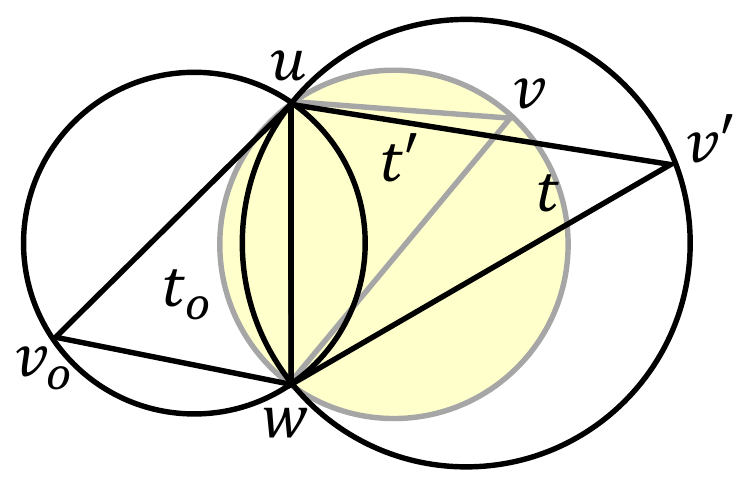}
\end{center}\vspace{-.5em}
\caption{An illustration of the procedure of \repboundary{}$(f,v)$ on
  a new point $v$ in two dimensions.    In this case the boundary face $f$ is $(u,w)$.
  The function will detach $\tri$ and replace it with $\tri'$.   The new triangle $\tri'$ only depends on $\tri_o$ and $\tri$.
  In support of Fact~\ref{lem:onlytwo}, we note
  that since $v$ encroaches on $\tri$ but not $\tri_o$, it
  must be in the larger black ball but not the smaller one.
  Therefore the yellow ball must be contained
  in the union of the two black balls and must contain the intersection
  of them.}\label{fig:delaunay}
\end{figure}

Our sequential variant of BT is described in
Algorithm~\ref{alg:delaunayseq}.  For each \dsimplex{} $\tri \in M$, the algorithm
maintains the set of uninserted points that encroach on $\tri$,
denoted as $E(\tri)$.  On each \step{}~$i$, the algorithm identifies the
boundary of the region that point $i$ encroaches on, and for each
face $f$ of that region it detaches the \dsimplex{} on the inside and replaces it
with a new \dsimplex{} $\tri'$ consisting of $f$ and $v$.  All work on
uninserted points is done in determining $E(\tri')$, which only
requires going through two existing sets, $E(\tri)$ and
$E(\tri')$.  This is justified by Fact~\ref{lem:onlytwo}.  Determining
the boundary of the region can be implemented efficiently by
maintaining a mapping from each point to the simplices it encroaches,
and checking those simplices.

\begin{fact} [\cite{BT93}]
  Given two $d$-simplices $\tri$ and $\tri_o$ that share a face $f$,
  and a point $v$ that encroaches on $\tri$ but not $\tri_o$, then for
  $\tri' = (f,v)$ we have
  $E(\tri) \cap E(\tri_o) \subseteq E(\tri') \subseteq E(\tri) \cup E(\tri_o)$.
\label{lem:onlytwo}
\end{fact}
This fact is proven in~\cite{BT93}, and an illustration of it is given in
Figure~\ref{fig:delaunay}.
A time bound for \seqinc{} of $O(n \log n + n^{\lceil d/2 \rceil})$
follows from the analysis of Boissonnat and Teillaud~\cite{BT93}.
Later we show a more precise bound on the number of \incircle{} tests
for $d = 2$, giving an upper bound on the constant factor for the
dominant term.

\subsection*{Dependence Depth and A Parallel Version}

We now consider the dependence depth of the algorithm.  One approach
is to consider dependences among the outer iterations (adding each
point).  Unfortunately it seems difficult to prove a logarithmic bound
on dependence depth for such an approach.  The problem is that
although a point will encroach on a constant number of triangles and
associated points in expectation, in some cases it could encroach on
up to a linear number.  Hence it does not have a $k$ bounded
dependence.  It seems that although expectation is good enough for the
work bound, it does not suffice for the depth bound since we need to
consider maximum depth over multiple paths.

We therefore consider a more fine-grained dependence structure based
on the triangles (\substep{}s) instead of points (top level
iterations).  The observation is that not all triangles added by a
point $v$ need to be added on the same round.  This will allows us to
show that Algorithm~\ref{alg:delaunayseq} has $k$ bounded nested
dependences.  We will make use of the following Lemma.

\begin{lemma}
  \label{lem:encroach}
  Consider two triangles $\tri$ and $\tri_o$ created by
  Algorithm~\ref{alg:delaunayseq} and sharing a face $f$ at some \step{}
  of the algorithm.  If the earliest point $v$ that encroaches on
  $\tri$ is earlier than any point that
  encroaches on $\tri_o$, then the algorithm will apply
  \repboundary$(\tri_o,f,\tri, v)$.
\end{lemma}
\begin{proof}
  Firstly, $v$ must be later than the points defining $\tri_0$ otherwise
  $\tri$ is detached before $\tri_0$ is created and the two never share a face.
  Once $\tri$ and $\tri_o$ are created the only points that can remove
  them are points that encroach on the triangles.  Since $v$ is the
  earliest such point, and only encroaches on $\tri$, when running the
  \step{} that inserts $v$, the triangles $\tri_o$ and $\tri$ will still
  be there, and \repboundary$(\tri_o,f,\tri,v)$ will be applied.
\end{proof}


We can now define a dependece graph $\gt(\points) = (T,E)$ based on
Lemma~\ref{lem:encroach}.  The vertices $T$ corresponds to triangles created by
Algorithm~\ref{alg:delaunayseq} (each sub-iteration), and for each
call to \repboundary$(\tri_o,f,\tri,\point)$ we include an arc from
each of $\tri_o$ and $\tri$ to the new triangle it creates $\tri'$.

\begin{theorem}
  Algorithm~\ref{alg:delaunayseq} has iteration dependence depth $O(d
  \log n)$ \whp{} over the random orders of $V$, i.e. $D(\gt(\points))
  = O(d \log n)$ \whp{}.
\label{the:delaunaydepth}
\end{theorem}
\begin{proof}
This follows from Theorem~\ref{thm:type1}.  In particular the
algorithm has a $2(d+1)$-bounded nested dependence.  Each creation of a
triangle (sub-iteration) by a point $v$ depends on at most two previous
triangles, each of which depends on at most $d+1$ points (the corners
of a triangle).  Therefore adding the triangle for point $v$ depends
on $2 \cdot (d+1)$ possible previous subiterations.  It is important
to note that in a given run there will only be dependences to the two
triangles, but the definition of $k$-bounded dependence requires we
consider all possible dependences to sub-iterations and associated
elements (points).  The tail count (the number of possible
sub-iterations that ended a dependence chain) is bounded by the number
of triangles in the result, $O(n^{\lceil d/2 \rceil})$, times the
number of points that could have generated each triangle, at most
$(d+1)$, giving a total of $O(d n^{\lceil d/2 \rceil})$.  Plugging into
Theorem~\ref{thm:type1}, gives a dependence depth of
\[\Pr(D(G) \geq \sigma H_n) < d n^{-(\sigma - (d+1)/2)} \]
for $\sigma \geq 2(d+1) e^2$, satisfying the bounds.
\end{proof}

\begin{algorithm}[t]
\caption{\parinc{}}
\label{alg:delaunaypar}
    \SetKwProg{myfunc}{function}{}{}
    \SetKwFor{ParForEach}{parallel foreach}{do}{endfch}
\KwIn{A sequence $\points = \{\point_1,\ldots,\point_n\}$ of points in $\R^d$.}
\KwOut{DT($\points$).}%
\KwMaintains{$E(\tri)$, the points that encroach on each \dsimplex{} $\tri$.}
  \vspace{.5em}
$\tri_b \gets$ a sufficiently large bounding \dsimplex{}\\
$E(\tri_b) \gets \points$\\
$M \gets \{t_b\}$ \\
\While {\upshape $E(\tri) \neq \varnothing$ for any $\tri \in M$} {
  \ParForEach {\upshape $(\tri_o,\tri)$ sharing a face $f \in M$,\\
    \hspace*{.9in}s.t. $\min(E(\tri)) < \min(E(\tri_o)$ \label{line:tritest}} {
      \repboundary{}$(t_o,f,t,\min(E(\tri)))$\label{line:reptri}
    }
}
\Return {$M$}
\end{algorithm}

Algorithm~\ref{alg:delaunaypar} describes a parallel variant of
Algorithm~\ref{alg:delaunayseq} based on the dependence structure.  On
each round the parallel algorithm applies
\repboundary{}$(\tri_o,f,\tri,\min(E(\tri))$ to all faces that satisfy the
conditions of Lemma~\ref{lem:encroach}---$\tri$ and $\tri_o$ are
present, and $\min(E(\tri)) < \min(E(\tri_o))$.  We assume points
are indexed by their insertion order, such that $\min(E(\tri))$
returns the earliest point and comparing two indices compares their
insertion order.
The subroutine \repboundary{} is
unchanged.  Because of Lemma~\ref{lem:encroach}, the parallel variant
will make exactly the same calls to \repboundary{} as the sequential
variant, just in a different order.
We note that since the triangles for a given point can be added on
different rounds, the triangulation is not necessarily self consistent
after each round.  Importantly, a face might only have one adjacent
triangle.  In that case the face cannot proceed until it receives the
second triangle (or is the boundary of the DT).  Also the faces of a
triangle can be detached on different rounds.  This does not affect
the algorithm---once all boundary faces of a point have been replaced,
the old interior will be fully detached from the new triangulation.

To implement the algorithm one can maintain three data structures: (1)
the set of triangles that have been created, each with the set of
points that encroach on it, (2) a hashmap that maps faces to their up
to two neighboring triangles, and (3) the set of faces that satisfy
the condition on line Line~\ref{line:tritest}, which we refer to as
the active faces.  The hashmap is indexed on the $d$ corners of a face
in some canonical order.  Each round goes over all the active faces in
parallel, and runs \repboundary{}.  This involves first looking up the
neighboring triangles, running the incircle tests across their points
in parallel, and filtering out the ones that return true.  The
algorithm also finds the minimum indexed such point.  Then the new
triangle is added to the triangle set, the $d+1$ faces of the new
triangle are updated in the hashmap (some might be new), and the
subset of them that satisfy Line~\ref{line:tritest} are added to the
set of active faces.

Most steps are easily parallelizable.  Applying and filtering on the
\incircle{} tests, and allocating the new active faces for each
\repboundary{}, can use processor allocation and compaction.  This can
be done approximately---i.e., into a constant factor larger set of
locations.  On the CRCW PRAM the approximate version can be be
implemented work efficiently in $O(\log^* n)$ depth
\whp{}\,\cite{Gil91a}.  On the CRCW PRAM the hash table operations and
the minimum can also can be done work efficiently in $O(\log^* n)$
depth \whp{}\,\cite{Gil91a,Hagerup91}.


\begin{theorem}
  \parinc{} (Algorithm~\ref{alg:delaunaypar}) runs in
  $O(d \log n \log^* n)$ depth \whp{}, and with work
  $O(n \log n + n^{\lceil d/2 \rceil})$ in expectation, on the CRCW PRAM.
\label{the:delaunaypar}
\end{theorem}
\begin{proof}
The number of rounds of \parinc{} is $D(\gt(\points))$ since the
iteration dependence graph is defined by the dependences in the
algorithm.  Each round has depth $O(\log^* n)$ \whp{} as described
above, so the overall depth is as stated.  The work of the algorithm
is the same as the sequential work~\cite{BT93} since the calls to
\repboundary{} are the same, and all steps are work-efficient.
\end{proof}

\subsection*{Work bound for $d = 2$}  We now derive a work bound on
the number of \incircle{} tests in two dimensions that includes the
constant factor on the high-order term.  In the proof we take
advantage that, due to Fact~\ref{lem:onlytwo}, the \incircle{} test is
not required for points that appear in both $E(\tri_o)$ and $E(\tri)$
since they will always appear in $E(\tri')$.  We know of no previous
work that gives this bound.

\begin{theorem}
\seqinc{} for $d=2$ and on $n$ points in random order does at most $24 n \ln n +
O(n)$ \incircle{} tests in expectation.
\end{theorem}
\begin{proof}

  We denote the point added at \step{} $i$ as $x_i$.  For an iteration
  $j$ we consider the history of \step{}s $i < j$, and we are interested
  in the ones that do \incircle{} tests on $x_j$.  For each such \step{}
  we consider the boundary of the region that $x_j$ encroaches
  immediately after \step{} $i$.  We will bound the number of \incircle{}
  tests on point $x_j$ based on the changes to this boundary over the
  \step{}s.  We define each face of the boundary by its two endpoints
  $(u,w)$ along with the (up to) two points sharing a triangle with
  $(u,w)$, which we denote as the four tuple $(u,w;v_l,v_r)$, and
  refer to as a \emph{winged edge}.  For example in Figure~\ref{fig:delaunay}
  the winged edge $(u,w;v_o,v)$ corresponds to the edge $(u,w)$ after
  adding $v$.

  In \repboundary{} a point is only tested for
  encroachment (an \incircle{} test) if its boundary winged edge
  $(u,w;v_l,v_r)$ is being deleted, and replaced with
  another.    This is because a point only needs to be tested if it
  encroaches on one side (one wing) and not the other.
  It seems to be messy to keep track of the deletions, however, so
  instead we keep track of additions of these boundaries.  We can then
  charge each deletion against the addition---i.e., we do the
  \incircle{} test on the deletion, but ``pay'' for it earlier on the
  addition.  This means we have to include some charge for the
  initial additions at the start of the algorithm.  This is $3$ per
  point, one for each edge of the bounding triangle.   However, when we add $x_j$ it has at least $3$ boundaries
  we don't have to pay for, so the net additional tests needed for this accounting
  method is at most zero.
  Let $Y_{ij}$ be the random variable specifying the number of
  boundaries for point $x_j$ that \step{} $i$ adds (i.e., winged
  edges that include $x_i$, and are on the boundary of $x_j$'s encroached
  region when added).   The total number of \incircle{} tests $C$ is then
bounded by \[C \leq \sum_{j=2}^n \sum_{i=1}^{j-1} Y_{ij}~.\]

  To analyze the expectation $E[Y_{ij}]$ we note that we can consider
  point $x_j$ as immediately following \step{} $i$ (since no other point
  $x_k, k > i$ has been added yet).  All points $x_1,\ldots,x_i,x_j$
  are equally likely to be selected as $x_j$, so the expected number
  of boundaries for $x_j$ is at most 6 (due to the fact that planar
  graphs can have average degree at most 6).  Each boundary winged
  edge has 4 points that could create it, any of which could be at
  position $i$.  Therefore $E[Y_{ij}]$, is upper bounded by the at
  most 6 boundaries in expectation, times the at most four points
  (worst case) and divided by the $i$ possible points
  $x_1,\ldots,x_i$, each equally likely.  This gives $E[Y_{ij}] \leq 6
  \times 4 /i = 24/i$, leading to the claimed result:
\vspace{-.1in}
\[\mathbb{E}[C] \leq \sum_{j=2}^n \sum_{i=1}^{j-1} E[Y_{ij}] = \sum_{j=2}^n
\sum_{i=1}^{j-1} \frac{24}{i} \leq 24 n \ln n + O(n).\]
\end{proof}

We note that it is easier to prove a looser $36n \ln n+ O(n)$ bound.
Basically every point encroaches on 4 triangles in expectation on each
\step{}, and each triangle has 3 points.  Now each of these triangles can
involve one, two, or three in-circle tests for an encroaching point
when is removed (depending on how many of its edges are on the
boundary of the encroached region).  This gives at most an expected
$3 \times 4 \times 3 / i$ per \step{} $i$ leading to the
$36n \ln n+ O(n)$ upper bound.  This is similar to the proof given by
GKS~\cite{GKS92} and appearing in some textbooks~\cite{BCKO}.

\section{Linear-Work Algorithms (Type 2)}
\label{sec:type2}
In this section, we study several problems from low-dimensional
computational geometry that have linear-work randomized incremental
algorithms. These algorithms fall into the Type 2 category of
algorithms defined in Section~\ref{sec:incalg2}, and their iteration
depth is polylogarithmic \whp{}. To obtain
linear-work parallel algorithms, we process the \step{}s in prefixes, as
described in Section~\ref{sec:incalg2}.
For simplicity, we describe the algorithms for
these problems in two dimensions, and and briefly note how they can
be extended to any fixed number of dimensions.

\subsection{Linear Programming}\label{sec:lp}
Constant-dimensional linear programming (LP) has received significant
attention in the computational geometry literature, and several
parallel algorithms for the problem have been
developed~\cite{Chen2002,Deng1990,Goodrich97,Dyer1995,Ajtai1992,Goodrich1996,Sen95,Alon1994}.
We consider linear programming in two dimensions. We assume that the
constraints are given in general position and the solution is either
infeasible or bounded. We note that these assumptions can be removed
without affecting the asymptotic cost of the
algorithm~\cite{Seidel1991}.  Seidel's~\cite{Seidel1991} elegant and
very simple randomized incremental algorithm adds the constraints
one-by-one in a random order, while maintaining the optimum point at
any time.  If a newly added constraint causes the optimum to no longer
be feasible (a tight constraint), we find a new feasible optimum point
on the line corresponding to the newly added constraint by solving a
one-dimension linear program, i.e., taking the minimum or maximum of
the set of intersection points of other earlier constraints with the
line. If no feasible point is found, then the algorithm reports the
problem as infeasible.

The iteration dependence graph is defined with the constraints as
\step{}s, and fits in the framework of Type 2 algorithms from
Section~\ref{sec:incalg2}. The \step{}s corresponding to inserting a
tight constraint are the special \step{}s.  Special \step{}s depend on all
earlier \step{}s because when a tight constraint executes, it needs to
look at all earlier constraints. Non-special \step{}s depend on the
closest earlier special \step{} $i$ because it must wait for \step{} $i$ to
execute before executing itself to retain the sequential execution (we
can ignore all of the earlier constraints since $i$ will depend on
them). Using backwards analysis, a \step{} $j$ has a probability of at most $2/j$ of being a
special \step{} because the optimum is defined by at most two constraints
and the constraints are in a randomized order.
Furthermore, the probabilities (event of being a special \step{}) are independent among different \step{}s.

As described in the proof of Theorem~\ref{thm:type2}, our parallel
algorithm executes the \step{}s in prefixes. Each time a prefix is
processed, it checks all of the constraints and finds the earliest one
that causes the current optimum to be infeasible using line-side
tests. The check per \step{} takes $O(1)$ work and processing a violating
constraint at \step{} $i$ takes $O(i)$ work and $O(1)$ depth \whp{} to
solve the one-dimensional linear program which involves
minimum/maximum operations.  Applying Theorem~\ref{thm:type2} with
$d(n) = O(1)$ gives the following theorem.

\begin{theorem}
Seidel's randomized incremental algorithm for 2D linear programming has
iteration dependence depth $O(\log n)$ and can be
parallelized to run in $O(n)$ work in expectation and $O(\log n)$
depth \whp{} on an arbitrary-CRCW PRAM.
\end{theorem}

We note that the algorithm can be extended to the case where the
dimension $d$ is greater than two by having a randomized incremental
$d$-dimensional LP algorithm recursively call a randomized incremental
algorithm for solving $(d-1)$-dimensional LPs.  This increases the
iteration dependence depth (and hence the depth of the algorithm) to
$O(d!\log^{d-1}n)$ \whp{}.  The work bound is $O(d!n)$ as in the
sequential algorithm~\cite{Seidel1991}.  We note that although the
work is optimal in $n$ it is not as good as the best sequential or
parallel algorithms~\cite{Alon1994} as a function of $d$, but is very
much simpler.



\subsection{Closest Pair}\label{sec:cp}
The \defn{closest pair} problem takes as input a set of points in the
plane and returns the pair of points with the smallest distance
between each other. We assume that no pair of points have the same
distance. A well-known expected linear-work
algorithm~\cite{Rabin76,Khuller1995,Golin1995,Har-peled2011} works by
maintaining a grid and inserting the points into the grid in a random
order. The grid partitions the plane into regions of size $r \times r$
where each non-empty region stores the points inside the region and
$r$ is the distance of the closest pair so far (initialized to the
distance between the first two points). It is maintained using a hash
table. Whenever a new point is inserted, one can check the region the
point belongs in and the eight adjacency regions to see whether the new
value of $r$ has decreased, and if so, the grid is rebuilt with the
new value of $r$. The check takes $O(1)$ work as each region can
contain at most nine points, otherwise the grid would have been rebuilt
earlier. Therefore insertion takes $O(1)$ work, and rebuilding the
grid takes $O(i)$ work where $i$ is the number of points inserted so
far. Using backwards analysis, one can show that point $i$ has
probability at most $2/i$ of causing the value of $r$ to decrease, so
the expected work is $\sum_{i=1}^nO(i)\cdot (2/i) = O(n)$.

This is a Type 2 algorithm, and the iteration dependence graph is
similar to that of linear programming. The special \step{}s are the
ones that cause the grid to be rebuilt, and the dependence depth is
$O(\log n)$ \whp{}.
Rebuilding the grid involves hashing, and
can be done in parallel in $O(i)$ work and $O(\log^*i)$ depth
\whp{} for a set of $i$ points~\cite{Gil91a}. We also assume that the
points in each region are stored in a hash table, to enable efficient
parallel insertion and lookup in linear work and $O(\log^*i)$ depth.
To obtain a linear-work parallel algorithm, we
again execute the algorithm in prefixes. Applying Theorem~\ref{thm:type2} with $d(n) = O(\log^*n)$ gives the
following theorem.

\begin{theorem}
The randomized incremental algorithm for closest pair can be
parallelized to run in $O(n)$ work in expectation and $O(\log n\log^*n)$
depth \whp{} on an arbitrary-CRCW PRAM.
\end{theorem}

We note that the algorithm can be extended to $d$ dimensions where the
depth is $O(\log d\log n\log^*n)$ \whp{} and expected work is $O(c_dn)$
where $c_d$ is some constant that depends on $d$.

\subsection{Smallest Enclosing Disk}\label{sec:sed}
The \defn{smallest enclosing disk} problem takes as input a set of
points in two dimensions and returns the smallest disk that contains
all of the points. We assume that no four points lie on a
circle. Linear-work algorithms for this problem have been
described~\cite{Megiddo83,Welzl1991}, and in this section we will
study Welzl's randomized incremental algorithm~\cite{Welzl1991}. The
algorithm inserts the points one-by-one in a random order, and
maintains the smallest enclosing disk so far (initialized to the
smallest disk defined by the first two points).  Let $v_i$ be the
point inserted on the $i$'th iteration.  If an inserted point $v_i$
lies outside the current disk, then a new smallest enclosing disk is
computed. We know that $v_i$ must be on the smallest enclosing
disk. We first define the smallest disk containing $v_1$ and $v_i$,
and scan through $v_2$ to $v_{i-1}$, checking if any are outside the
disk (call this procedure \textbf{Update1}). Whenever $v_j$ ($j<i$) is
outside the disk, we update the disk by defining the disk containing
$v_i$ and $v_j$ and scanning through $v_1$ to $v_{j-1}$ to find the
third point on the boundary of the disk (call this procedure
\textbf{Update2}). \textbf{Update2} takes $O(j)$ work, and
\textbf{Update1} takes $O(i)$ work plus the work for calling
\textbf{Update2}. With the points given in a random order, the
probability that the $j$'th iteration of \textbf{Update1} calls
\textbf{Update2} is at most $2/j$ by a backwards analysis argument, so
the expected work of \textbf{Update1} is $O(i) +
\sum_{j=1}^i(2/j)\cdot O(j) = O(i)$. The probability that
\textbf{Update1} is called when the $i$'th point is inserted is at
most $3/i$ using a backwards analysis argument, so the expected work
of this algorithm is $\sum_{i=1}^n(3/i)\cdot O(i) = O(n)$.

This is another Type 2 algorithm whose iteration dependence graph is
similar to that of linear programming and closest pair.  The points
are the \step{}s, and the special \step{}s are the ones that cause
\textbf{Update1} to be called, which for \step{} $i$ has at most $3/i$
probability of happening. The dependence depth is again $O(\log n)$
\whp{} as discussed in Section~\ref{sec:incalg2}.


Our work-efficient parallel algorithm again uses prefixes, both when
inserting the points, and on every call to \textbf{Update1}.  We
repeatedly find the earliest point that is outside the current disk by
checking all points in the prefix with an in-circle test and taking
the minimum among the ones that are outside. \textbf{Update1} is
work-efficient and makes $O(\log n)$ calls to \textbf{Update2} \whp{},
where each call takes $O(1)$ depth \whp{} as it does in-circle tests
and takes a maximum.  As in the sequential algorithm, each \step{} takes
$O(1)$ work in expectation.
Applying Theorem~\ref{thm:type2} with
$d(n) = O(\log n)$ \whp{} (the depth of a executing a \step{} and calling
\textbf{Update1}) gives the following theorem.

\begin{theorem}
The randomized incremental algorithm for smallest enclosing disk can be
parallelized to run in $O(n)$ work in expectation and $O(\log^2 n)$
depth \whp{} on an arbitrary-CRCW PRAM.
\end{theorem}

The algorithm can be extended to $d$ dimension,
with $O(d!\log^dn)$ depth \whp{}, and
$O(c_dn)$ expected work for some constant $c_d$ that depends on $d$.
Again, we can use the same randomized order
for all sub-problems.

\section{Iterative Graph Algorithms (Type 3)}
\label{sec:graph}

In this section we study two sequential graph algorithms that can be
viewed as offline versions of randomized incremental algorithms.  We show
that the algorithms are Type 3 algorithms as described in
Section~\ref{sec:incalg3}, and also that \step{}s executing in parallel
can be combined efficiently. This gives us simple
parallel algorithms for the problems.
The algorithms use single-source shortest paths and reachability as
(black-box) subroutines, which is the dominating cost.
Our algorithms are within a logarithmic factor in work and depth
of a single call to these subroutines on the input graph.



\subsection{Least-Element Lists}\label{sec:LE-lists}
The concept of Least-Element lists (LE-lists) for a graph (either
unweighted or with non-negative weights) was first proposed by
Cohen~\cite{cohen1997} for estimating the neighborhood sizes of
vertices. The idea has subsequently been used in many applications
related to estimating the influence of vertices in a network
(e.g.,~\cite{cohen2004,du2013} among many others), and generating
probabilistic tree embeddings of a graph~\cite{khan2012,FRTnew} which
itself is a useful component in a number of network optimization
problems and in constructing distance oracles~\cite{FRTnew,kanat}.
For $d(u,v)$ being the shortest path from $u$ to $v$ in $G$, we have:

\begin{definition}[LE-list]
Given a graph $G=(V,E)$ with $V = \{v_1, \ldots, v_n\}$, the \defn{LE-lists}
are
\[L(v_i)=\left\{v_j \in V~|~ d(v_i,v_j) < \min_{k=1}^{j-1} d(v_i,v_k)\right\}\]
sorted by $d(v_i,v_j)$.
\end{definition}

\begin{algorithm}[t]
\caption{The iterative LE-lists construction~\cite{cohen1997}}
\label{algo:lelist}
\KwIn{A graph $G=(V,E)$  with $V = \{v_1,\ldots,v_n\}$}
\KwOut{The LE-lists $L(\cdot)$ of $G$}
    \medskip
    Set $\delta(v)\leftarrow +\infty$ and $L(v) \leftarrow \varnothing$ for all $v\in V$\\
    \For{\upshape $i\leftarrow 1$ to $n$} {
       Let $S=\{u \in V~|~d(v_i,u)<\delta(u)\}$\label{stp:dijkstra}\\
       \For{$u \in S$}{
           $\delta(u)\leftarrow d(v_i,u)$ \\
           concatenate $\left<v_i,d(v_i,u)\right>$ to the end of $L(u)$
      }
    }
    \Return{$L(\cdot)$}
\end{algorithm}

In other words, a vertex $u$ is in vertex $v$'s LE-list if and only if
there are no earlier vertices (than $u$) that are closer to $v$.
Often one stores with each vertex $v_j$ in $L(v_i)$ the distance of
$d(v_i,v_j)$.

Algorithm~\ref{algo:lelist} provides a sequential iterative (incremental)
construction of the LE-lists, where the $i$'th \step{} is the $i$'th
iteration of the for-loop.  The set $S$ captures all vertices that
are closer to the $i$'th vertex than earlier vertices
(the previous closest distance is stored in $\delta(\cdot)$).
Line~\ref{stp:dijkstra} involves computing $S$ with a single-source
shortest paths (SSSP) algorithm (e.g., Dijkstra's algorithm for weighted graphs and BFS for unweighted graphs, or other algorithms~\cite{spencer1997time,ullman1991high,klein1997randomized,Blelloch16shortest} with more work but less depth). We note that the only minor change to these algorithms is to drop the initialization of the tentative distances before we run SSSP, and instead use the $\delta(\cdot)$ values from previous \step{}s in Algorithm~\ref{algo:lelist}.
Thus the search will only explore $S$ and its outgoing edges.
Cohen~\cite{cohen1997} showed that if  the vertices are in random order, then each
LE-list has size $O(\log n)$ \whp{}, and that using Dijkstra
with distances initialized with $\delta(\cdot)$,
the algorithm runs in $O(\log n (m + n \log n))$ time.

\myparagraph{Parallel version}  To parallelize the algorithm we use
the general approach of Type 3 algorithms as described in
Section~\ref{sec:incalg3} and in particular
Algorithm~\ref{alg:type3}.  We treat the shortest paths algorithm as a
black box that computes the set $S$ in depth $D_{\smb{SP}}(n',m')$ and
work $W_{\smb{SP}}(n',m')$, where $n'=|S|$ and $m'$ is the sum of the
degrees of $S$.   We assume the cost functions are concave, i.e.
$W_{\smb{SP}}(n_1,m_1)+W_{\smb{SP}}(n_2,m_2)\le W_{\smb{SP}}(n,m)$ for
$n_1+n_2\le n$ and $m_1+m_2\le m$, which holds for all existing
shortest paths algorithms.  We also assume independent shortest path
computations can run in parallel (i.e., they do not intefere with each
other's state).   We assume that the output of each shortest
path computation from a source is a set of source-target-distance
triples, one for each target that is visited in Line~3.

For the separating dependences: $v_j$ depends on $v_i$ if and only if
$v_i \in L(v_j)$, (i.e., was searched by $v_i$), and we use the total
orderings $i\less_k j$ if $d(v_k,v_i)<d(v_k,v_j)$.  This gives:

\begin{lemma}\label{lem:lelist}
Algorithm~\ref{algo:lelist} has a separating dependence for the
dependences and orderings $<_v$ defined above.
\end{lemma}

\begin{proof}
By Definition~\ref{def:sepdep}, we need to show that
for any three vertices $v_a, v_b, v_c \in V$, if $v_a \less_c v_b \less_c
v_c $ or $v_c \less_c v_b \less_c v_a$, then $v_c$ can only be visited on $v_a$'s \step{} if
$v_a$'s \step{} is the first among the three.

Clearly the statement holds if $v_c$'s \step{} is the earliest among the
three.  We now consider the case when $v_b$'s \step{} is the first among
the three.  Since $d(v_c,v_c)=0$, $v_a \less_c v_b \less_c v_c$ cannot
happen, so we only need to consider the case $v_c \less_c v_b \less_c
v_a$.  Since $d(v_c,v_b)<d(v_c,v_a)$ and $b<a$, based on the
definition of the LE-lists, $v_a\notin L(v_c)$.  As a result, $v_c$
can only be visited in $v_a$'s \step{} if $v_a$'s \step{} is first among the
three.
\end{proof}

As required by Line~\ref{line:combine} of Algorithm~\ref{alg:type3},
we need to combine the results from the iterations in a
round---the sets of source-target-distance triples.  For LE-lists we
need to collect the contributions to each LE-list, remove the
redundant entries, and write the minimum distance to each vertex for
the next round.  There are redundant entries since running an
iteration early could find a path not found by the strict sequential
order.  Collecting the contributions to each LE-list can be done with
a semisort on the targets.  The elements corresponding to each
target can then be sorted based on the iteration number of the source
vertex.  In the sequential order, distances can only decrease with
increasing source iteration index. 
Therefore, if any of the distances increase, they
correspond to redundant entries that are filtered out.  Finally, the
remaining elements are added to the appropriate LE-lists and the
minimum distance is written to each vertex.  This leads to the
following theorem:

\begin{theorem}
The LE-lists of a graph with the vertices in random order can be constructed in
$O(W_{\smbe{SP}}(n,m)\log n)$ expected work and
$O(D_{\smbe{SP}}(n,m)\log n)$ depth \whp{} on the CRCW PRAM.
\end{theorem}
\begin{proof}
First we bound the cost of the algorithm excluding the post-processing
step.
Because of the separating dependences in Algorithm~\ref{algo:lelist} shown in Lemma~\ref{lem:lelist},
Theorem~\ref{thm:type3} indicates that each vertex is visited no more than $O(\log n)$ times in all iterations \whp{}, assuming a random input order of the vertices.
Namely, at most $O(\log n)$ searches visit each vertex and its neighbors.
Since we assume the concavity of the search cost, the overall work for all searches is $O(\log n)$ times $W_{\smbe{SP}}(n,m)$, the cost of the first search that visits all vertices.
\hide{, we can apply Theorem~\ref{thm:type3} and
Lemma~\ref{lem:lelist} with the convexity of the work complexity, and
this is within the claimed bounds.
\guy{more detail}}

The combining after each round requires a semisort on the target vertex, a sort
on the source vertex within each target, and a pass to remove duplicates.  The
semisort can be done in linear work and logarithmic depth~\cite{RR89,GSSB15}.
We now show that we can efficiently sort by source.
As shown in the proof of Theorem~\ref{thm:type3}, the probability that we need to sort $r$ elements for each vertex in one round is bounded by $2^{1-r}$.
Assume that we use a loose upper bound of quadratic work for sorting.
The expected work for sorting the elements for each vertex in one round is $\sum{2^{1-r}\cdot r^2}$ for $r\ge 1$, which solves to $O(1)$.
The only exception is that the vertex itself is always in its own LE-list.
This adds at most $O(\log n)$ to the work to compare it to all
elements in the list.  Thus the expected work to sort one LE-list in
expectation is $O(\log n)$, and is $O(n\log n)$ when summed across
all lists.
The work cost for sorting is dominated since we assume that $W_{\smbe{SP}}(n,m)$ is at least linear, and we need $O(\log n)$ reachability queries.
The depth for sorting is $O(\log n)$ (from Section~\ref{sec:sorting} or~\cite{Cole88}), which is
within the claimed bounds.
\end{proof}

\newcommand{\vsub}{{\mathcal V}}
\newcommand{\vscc}{{V_{\text{\emph{scc}}}}}
\newcommand{\sscc}{{S_{\text{\emph{scc}}}}}

\subsection{Strongly Connected Components}\label{sec:SCC}
Given a directed unweighted graph $G=(V,E)$, a \defn{strongly connected component}
(SCC) is a maximal set of vertices $C\subseteq V$ such that for every
pair of vertices $u$ and $v$ in $C$, there are directed paths both
from $u$ to $v$ and from $v$ to $u$.  Tarjan's
algorithm~\cite{tarjan1972} finds all strongly connected components of
a graph using a single pass of depth-first search (DFS) in $O(|V|+|E|)$
work.  However, DFS is generally considered to be hard to
parallelize~\cite{reif1985}, and so a divide-and-conquer SCC
algorithm~\cite{coppersmith2003} is usually used in parallel
settings~\cite{hong2013fast,barnat2011computing,slota_ipdps2014_bfs,Tomkins2015}.

The basic idea of the divide-and-conquer algorithm is similar to
quicksort.  It applies forward and backward reachability queries for a
specific ``pivot'' vertex $v$, which partitions the remaining vertices
into four subsets of the graph, based on whether it is forward
reachable from $v$, backward reachable, both, or neither.  The subset
of vertices reachable from both directions form a strongly connected
component, and the algorithm is applied recursively to the three
remaining subsets.  Coppersmith et al.~\cite{coppersmith2003} show
that if the vertex $v$ is selected uniformly at random, then the algorithm
sequentially runs in $O(m \log n)$ work in expectation.

Although divide-and-conquer is generally good for parallelism, the challenge
in this algorithm is that the divide-and-conquer tree can be
very unbalanced.  For example, if the input graph is very sparse
such that most of the reachability searches only visit a few vertices,
then most of the vertices will fall into the subset of unreachable
vertices from $v$, creating unbalanced partitions with $\Theta(n)$
recursion depth.  Schudy~\cite{schudy2008finding} describes a technique to better balance the
partitions, which can bound the depth of the algorithm to be $O(\log^2 n)$ reachability queries.
Unfortunately, his approach requires a factor of $O(\log n)$ extra
work compared to the original algorithm, which is significant.
Tomkins et al.~\cite{Tomkins2015} describe another parallel approach, although
the analysis is quite complicated.\footnote{\small Tomkins et
   al.\,\cite{Tomkins2015} claim that their algorithm takes the same
   amount of work as the sequential algorithm, but it seems that
   there are errors in their analysis. For example, the goal of the analysis
   is to show that in each round their algorithm visits $O(n)$ vertices in expectation,
   which they claimed to imply visiting $O(m)$ edges in expectation.  This is not
   generally true since the vertices do not necessarily have the
   same probabilities of being visited.  Other than this, their work
   contains many interesting ideas that motivated us to look at this
   problem.}

The reason that we can design a simple algorithm and bound the depth of the recursion depth is based on the following intuition:
the divide-and-conquer algorithm~\cite{coppersmith2003} can also be
viewed as an incremental algorithm, and we describe this version in
Algorithm~\ref{algo:seq}.  The two versions are equivalent since a random ordering is equal to picking
the pivots at random.   Therefore, we can analyze it as a
Type 3 algorithm, using the general theorem shown in Section~\ref{sec:incalg3}.
Our analysis is significantly simpler than those
of~\cite{schudy2008finding,Tomkins2015}, and the asymptotic work of
our algorithm matches that of the sequential algorithm.

\begin{algorithm}[t]
\caption{The sequential iterative SCC algorithm}
\label{algo:seq}
\KwIn{A directed graph $G=(V,E)$ with $V = \{v_1,\ldots,v_n\}$.}
\KwOut{The set of strongly connected components of $G$.}
    \medskip
    $\vsub\gets\{\{v_1,v_2,\ldots,v_n\}\}$~~~~~~(Initial Partition)\\
    $\sscc\gets\{\}$\\
    \For{\upshape $i\leftarrow 1$ to $n$}
    {
        Let $S\in \vsub$ be the subgraph that contains $v_i$\\
        \lIf {$S=\varnothing$} {go to the next iteration}
        $R^+\leftarrow\mf{Forward-Reachability}(S,v_i)$\label{line:forward}\\
        $R^-\leftarrow\mf{Backward-Reachability}(S,v_i)$\label{line:backwards}\\
        $\vscc\leftarrow R^+\cap R^-$\label{line:getSCC}\\
        $\vsub\leftarrow \vsub \backslash \{S\}\cup\{R^+ \backslash \vscc, R^- \backslash \vscc, S\backslash (R^+\cup R^-)\}$\\
        $\sscc\leftarrow \sscc \cup \{\vscc\}$\label{line:addSCC}\\
     }
    \Return{$\sscc$}
\end{algorithm}

As in previous work on parallel SCC algorithms, we treat the algorithm
for performing reachability queries as a black box with $W_R(n,m)$
work and $D_R(n,m)$ depth, where $n$ are the number of reachable
vertices and $m$ is the sum of their degrees.  It can be implemented
using a variety of algorithms with strong theoretical
bounds~\cite{spencer1997time,ullman1991high} or simply with a
breadth-first search for low-diameter graphs.  We also assume
convexity on the work $W_R(n,m)$, which holds for existing
reachability algorithms, and that independent reachability
computations can run in parallel.

We first show that the algorithm has \sepdep{}.
Here a dependence from $i$ to $j$ corresponds to a forward or backward reachability search
from $i$ visiting $j$  (Lines~\ref{line:forward} and~\ref{line:backwards} in Algorithm~\ref{algo:seq}).  Let
$T=(t_1,t_2,\ldots,t_n)$ be an arbitrary topological order of components
in the given graph $G$, in which vertices of the same component are arbitrarily
ordered within the component.  $T$ is not constructed explicitly, but only used in
analysis.  To define the total order for vertex $v_i$, i.e.,
$<_{v_i}$, we
take all vertices of $T$ that are forward or backward reachable from $v_i$
(including $v_i$ itself) and
put them at the beginning of the ordering (maintaining their relative order), and put the unreachable vertices after them.
Given this ordering, we have the following lemma.

\begin{figure}[t]
\begin{center}
  \includegraphics[width=.5\columnwidth]{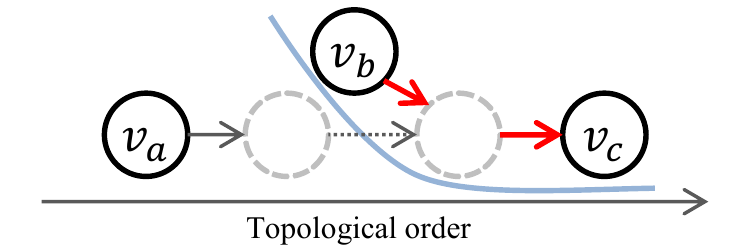}
\end{center}\vspace{-1em}
\caption{An illustration for the proof of Lemma~\ref{lem:scc}.}\label{fig:scc}
\vspace{-.5em}
\end{figure}

\begin{lemma}\label{lem:scc}
Algorithm~\ref{algo:seq} has a separating dependence for the
dependences and orderings $<_v$ defined above.
\end{lemma}

\begin{proof}
By Definition~\ref{def:sepdep}, we need to show that
for any three vertices $v_a, v_b, v_c \in V$, if $v_a \less_c v_b
\less_c v_c $ or $v_c \less_c v_b \less_c v_a$, $v_c$ can only be
reached (forward or backward) in $v_a$'s \step{} if $v_a$'s \step{} is the first
among the three.

Clearly the statement is true if $v_c$ is earliest.  We now consider
the case when $v_b$'s \step{} is first among the three vertices.  We give
the argument for the forward direction, and the backward direction is
true by symmetry.  If $v_b$ is in the same SCC as either $v_a$ or
$v_c$, then in $v_b$'s \step{}, either $v_a$ or $v_c$ is marked in one SCC that $v_b$ is in,
and removed from the subgraph set $\vsub$.  Otherwise, since $v_b$ and
$v_a$ are not in the same SCC, when $v_c <_c v_b <_c v_a$, $v_c$
cannot be forward reachable in $v_a$'s \step{}, and when $v_a \less_c v_b
\less_c v_c$, $v_a$ cannot be forward reachable from $v_b$'s \step{}.  In
the second case, after $v_b$'s \step{}, the forward reachability search
from $v_b$ reaches $v_c$ but not $v_a$, and so $v_a$ and $v_c$ fall into
different components in $\vsub$ (shown in Figure~\ref{fig:scc}).  As a
result, $v_c$ is also not reachable in $v_a$'s \step{}.

In conclusion, $v_c$ can only be
reached (forward or backward) in $v_a$'s \step{} if $v_a$'s \step{} is first
among the three.
\end{proof}

This separating dependence implies that the sequential
algorithm on a random ordering does $O(m \log n)$ work since each
vertex $v_j$ is visited by no more than $\sum_{i=1}^{j} 2/i = O(\log n)$ times in
expectation shown by Lemma~\ref{lem:sep}, and the upper bound of this expectation is
independent of the degrees.  Since $W_R(n,m) = O(m)$ sequentially,
e.g., using BFS, this algorithm uses $O(m \log n)$ work on expectation.

We now consider the parallel version.  We use the general approach of
Type 3 algorithms as described in Section~\ref{sec:incalg3} and in
particular Algorithm~\ref{alg:type3}.  The \step{}s we run in parallel
for each round (consisting of increasing power of two) are the same as
in Algorithm~\ref{algo:seq}.  To implement the parallel version, we
need a way to efficiently combine the \step{}s that run in the same
round.  Based on our assumption, the reachability queries for each
vertex in the round can run independently in parallel, and so we run
the searches based on the partitioning of the vertices $\vsub$ from
the previous round. For each direction, each search does a priority write with its ID in a temporary location to all the vertices it visits.
We can then identify each strongly connected
component by checking the reachability
information on vertices.
Vertices belonging in an SCC will have the ID of the highest priority search written in both of its temporary locations (one for each search direction), and vertices with the same ID written belong to the same SCC.
To partition the graph, any
edge between two vertices, where one is reachable and the other is
not, in any of the reachability queries is cut.  Each reachability
search can identify and cut its own edges (some edges might be cut
multiple times).  This implementation is more aggressive than the the
sequential algorithm, but this will only help.  If required, the exact
intermediate states of the sequential algorithm can also be
maintained.  After all $O(\log n)$ rounds are complete, we group the
vertices to form SCCs based on their vertex labels.  This requires
linear work and $O(\log n)$ depth~\cite{RR89,GSSB15}.



\begin{theorem}
For a random order of the input vertices, the incremental SCC algorithm does
$O(W_R(n,m)\log n)$ expected work and has $O(D_R(n,m)\log n)$ depth
on a priority-write CRCW PRAM.
\end{theorem}
\begin{proof}
The overall extra work is $O(m \log n)$ and no more than the work for executing the reachability queries.
The depth for the additional operations is constant in each round and $O(\log n)$ at the end of the algorithm.
Therefore if the input vertices are randomly permuted, we can apply
Theorem~\ref{thm:type3} with Lemma~\ref{lem:scc} and the
convexity of the work cost to bound the expected work and depth of the
algorithm.
\end{proof}

\myparagraph{Acquiring the same intermediate states as the sequential
  algorithm} The partitioning of the vertex sets in previously
discussed algorithm is more eager than the sequential algorithm.  When
determinism is needed, the same intermediate states in the sequential
algorithm can be retrieved in the parallel version.  Let's consider
the following case in one parallel round: vertex $z$ is forward
reached from $x$ and reached from $y$, and at the meantime $x$ has a
higher priority.  The search of $y$ affects $z$ if and only if $y$ is
also reached in $x$'s forward search.  The other direction is
symmetric.

With this observation, we can use the following algorithm to decide
the partitioning of the vertices after one round in the sequential
order.  After the searches in each round finish, we first check
whether each vertex is already in an SCC.  For the vertices not in an
SCC, we semisort pairs based on the source vertex of the search and
the reached vertex, and gather all searches that reach each vertex in
both directions.  For each vertex, we then sort these searches based
on the priorities of the source nodes, and use a scan to filter out
the searches that do not reach this vertex in the sequential
algorithm, based on the criteria above.  Then the partitions are
decided by the search with the lowest priority that reaches each
vertex.  Finally, we cut the edges based on the partitioning of the
vertex sets using $O(m)$ work and constant depth.  Using this approach
guarantees the same intermediate states of the partitions at the end of each round compared to the sequential algorithm.  The extra work in this step is the same as
the post-processing in LE-lists in Section~\ref{sec:LE-lists}, which
is dominated by the other parts in the algorithm.

\section{Conclusion}
In this paper, we have analyzed the dependence structure in a
collection of known randomized incremental algorithms (or slight
variants) and shown that there is inherently high parallelism in all of
the algorithms.  The approach leads to particularly simple parallel
algorithms for the problems---only marginally more complicated (if at
all) than some of the very simplest efficient sequential algorithms
that are known for the problems.  Furthermore the approach allows us
to borrow much of the analysis already developed for the sequential
versions (e.g., with regard to total work and correctness).  
We presented three general types of dependences of algorithms, and
tools and general theorems that are useful for multiple algorithms
within each type.
We expect that there are many other algorithms that can
be analyzed with these tools and theorems.  

\section*{Acknowledgments}
This research was supported in part by NSF grants
CCF-1314590 and CCF-1533858, the Intel Science and Technology Center
for Cloud Computing, and the Miller Institute for Basic Research in
Science at UC Berkeley.

\bibliographystyle{abbrv}

\end{document}